\documentclass[conference]{IEEEtran}
\usepackage[cmex10]{amsmath}

\usepackage{amssymb}
\usepackage{amsmath}%
\usepackage{amsfonts}%
\usepackage{amssymb}%
\usepackage{amsthm}
\usepackage{breqn}
\usepackage{graphicx}
\usepackage{colonequals}
\usepackage{psfrag}
\usepackage{subfigure}
\usepackage{mathrsfs}

\usepackage[usenames,dvipsnames]{xcolor}
\usepackage{mdframed}
\usepackage{bm}
\usepackage{bbm}
\usepackage{cite}

\usepackage{prettyref,enumerate}

\newtheorem{thm}{Theorem}
\newtheorem{lem}{Lemma}

\theoremstyle{definition}

\newtheorem{defn}{Definition}

\newtheorem{rem}{Remark}
\newtheorem{example}{Example}

\newcommand{\calX}{\mathcal{X}}

\newcommand{\calY}{\mathcal{Y}}

\newcommand{\calZ}{\mathcal{Z}}

\newcommand{\calU}{\mathcal{U}}

\newcommand{\Reals}{\mathbb{R}}

\newcommand{\defined}{\triangleq}

\IEEEoverridecommandlockouts
\begin{document}
\title{The Utility Cost of Robust Privacy Guarantees}

\author{\authorblockN{Hao Wang\authorrefmark{1}, Mario Diaz\authorrefmark{1}\authorrefmark{2}, Flavio~P.~Calmon\authorrefmark{1} and Lalitha Sankar\authorrefmark{2}}\authorblockA{\authorrefmark{1}Harvard University, \{hao\_wang,mdiaztor\}@g.harvard.edu, flavio@seas.harvard.edu}\authorblockA{\authorrefmark{2}Arizona State University, \{mdiaztor,lsankar\}@asu.edu}\thanks{This material is based upon work supported by the National Science Foundation under Grant No. CCF-1350914 and an ASU seed grant.}}

\maketitle
\begin{abstract}
Consider a data publishing setting for a data set with public and private features. The objective of the publisher is to maximize the amount of information about the public features in a revealed data set, while keeping the information leaked about the private features bounded. The goal of this paper is to analyze the performance of privacy mechanisms that are constructed to match the distribution learned from the data set. Two distinct scenarios are considered: (i) mechanisms are designed to provide a privacy guarantee for the learned distribution; and (ii) mechanisms are designed to provide a privacy guarantee for every distribution in a given neighborhood of the learned distribution. For the first scenario, given any privacy mechanism, upper bounds on the difference between the privacy-utility guarantees for the learned and true distributions are presented. In the second scenario, upper bounds on the reduction in utility incurred by providing a uniform privacy guarantee are developed.
\end{abstract}


\section{Introduction}

The disclosure of data with both privacy and utility guarantees is a recognized objective in many applications. A common approach to this problem is to process the data set through a \textit{privacy mechanism} that seeks to  fulfill certain privacy and utility guarantees. Information theoretic methods for designing privacy mechanisms often rely on the implicit assumption that the data distribution is, for the most part, known  \cite{du2012privacy,sankar2013utility,asoodeh2016information,wang2017estimation}. However, in practice, the data distribution may only be accessed through a limited number of observed samples.

In this work, we revisit this assumption, and study the robustness of privacy and utility guarantees of information-theoretic privacy mechanisms to partial knowledge of the input distribution. In practice, this inaccuracy stems from the  limited availability of samples which, in turn, produces a discrepancy between the learned and the true data distribution. To mitigate the effect of this discrepancy, we also study the performance of privacy mechanisms that, by design, are robust: they assure privacy for {\it every} data set drawn from a distribution within a neighborhood. Here, the neighborhood is given by an $\ell_1$-ball of radius $r\geq0$ around a distribution estimated from a limited number of samples. Our analysis can be applied when  privacy and utility  are measured in terms of a broad range of metrics based on $f$-divergences, or by probability of correct guessing.

Due to  its natural interpretation and simplicity, we start our analysis by letting $r=0$. This corresponds to the {\it pointwise} setting, where the privacy mechanism is fixed, and its performance is evaluated in terms of a single distribution learned from data. 
We provide bounds on the gap between the privacy-utility guarantees computed under the empirical distribution and the \textit{de facto} guarantees for the true data distribution. This gap depends on the number of observed samples and properties of the data  (e.g. support size, probability of least likely symbol), and improves and generalizes the results presented by Wang and Calmon in \cite{wang2017estimation}.

We then extend our analysis to the more general {\it uniform} setting. Here, a given level of privacy is uniformly assured for all data sets drawn from distributions within a neighborhood $r>0$ of a target distribution (potentially learned from data). Using large deviation results, we establish upper bounds on the reduction in utility due to the uniform privacy guarantee which, in turn, depend on the value of $r$. 

The paper is organized as follows. In Section~\ref{Subsection:PUTs} we recall the framework of privacy-utility trade-offs. The formal definitions for the uniform privacy guarantees are introduced in Section~\ref{Subsection:RobustPUT}. In Section~\ref{Subsection:EmpiricalDistribution} we recall basic results from large deviation theory related to the distance between the empirical and true distributions. The definitions of $f$-informations and probability of correct guessing are recalled in Section~\ref{Subsection:fInformationPc}. Our main results for the pointwise and uniform results are presented in Sections~\ref{Section:PointwiseRobustness} and \ref{Section:UniformRobustness}, respectively. 

\section{Problem Setup and Preliminaries}
\subsection{Privacy-Utility Trade-Offs}
\label{Subsection:PUTs}

Suppose that $S$ is a variable to be hidden (e.g. political preference) and $X$ is an observed variable  (e.g. movie ratings) that is correlated with $S$. In order to receive some utility (e.g. personalized recommendations), we would like to disclose as much information about $X$ without compromising $S$. An approach with rigorous privacy guarantees is to release a new random variable $Y$ produced by applying a randomized mapping to $X$. This mapping, called the privacy mechanism, is designed to fulfill a certain privacy constraint.

In the sequel we assume that $S$ and $X$ are discrete and let $P_{S,X}$ denote their joint distribution. The support of $Y$ can be any discrete set. We let $\mathcal{L}(P_{S,X},P_{Y|X})$ and $\mathcal{U}(P_{S,X},P_{Y|X})$ be the privacy leakage and the utility generated by a mapping $P_{Y|X}$ for the underlying distribution $P_{S,X}$, respectively. Throughout this paper, specific instantiations of $\mathcal{L}$ and $\mathcal{U}$ are $f$-informations and probability of correctly guessing. The following definition captures the fundamental trade-off between privacy and utility in the present setting.

\begin{defn}
\label{defn:PUT}
For a given joint distribution $P_{S,X}$ and $\epsilon\geq0$, the {\it privacy-utility function} is defined as
\begin{equation}
\mathsf{H}(P_{S,X};\epsilon) \defined \sup_{P_{Y|X}\in\mathcal{D}(P_{S,X};\epsilon)} \mathcal{U}(P_{S,X},P_{Y|X}),
\end{equation}
where $\mathcal{D}(P_{S,X};\epsilon) \defined \{ P_{Y|X} : \mathcal{L}(P_{S,X},P_{Y|X}) \leq \epsilon\}$.
\end{defn}

This type of privacy-utility trade-off (PUT) has been investigated for several measures of privacy and utility, see, for example, \cite{asoodeh2016information,asoodeh2017estimation,wang2017estimation}. When the distribution $P_{S,X}$ is known, the privacy-utility function in Definition \ref{defn:PUT} quantifies the best utility achievable by any privacy mechanism providing the desired privacy guarantee. In practice, the designer may not have access to the true distribution $P_{S,X}$, but only to independent samples $\{(s_i,x_i)\}_{i=1}^n$ drawn from this distribution. In this case, the privacy-utility guarantees for a distribution learned from the samples, say $\hat{P}_{S,X}$, and the true distribution $P_{S,X}$ might be different. For any given privacy mechanism $P_{Y|X}$, these discrepancies are effectively quantified by
\begin{equation}
|\mathcal{L}(P_{S,X},P_{Y|X}) - \mathcal{L}(\hat{P}_{S,X},P_{Y|X})|,
\end{equation}
and
\begin{equation}
|\mathcal{U}(P_{S,X},P_{Y|X}) - \mathcal{U}(\hat{P}_{S,X},P_{Y|X})|.
\end{equation}

\subsection{Uniform Privacy Guarantees}
\label{Subsection:RobustPUT}

When privacy is a priority, a specific privacy guarantee for the true distribution $P_{S,X}$ may still be required, even though the designer has only access to a distribution $\hat{P}_{S,X}$ estimated from the samples $\{(s_i,x_i)\}_{i=1}^n$. We propose the following procedure to overcome this difficulty: (a) use large deviation theory results to find an upper bound, say $r$, for the distance between $\hat{P}_{S,X}$ and $P_{S,X}$; (b) provide a privacy guarantee for {\it all} distributions at distance less or equal than $r$ from the $\hat{P}_{S,X}$. In the sequel, we measure the distance between two probability distributions $P$ and $Q$ by their $\ell_1$-distance,
\begin{equation}
    \|P - Q\| \defined \sum_{z\in\calZ} |P(z) - Q(z)|.
\end{equation}
With this notation, we introduce the following definition.

\begin{defn}
\label{def:robust-utility-privacy}
Given $\hat{P}_{S,X}$, $\epsilon\geq0$, and $r\geq0$, we define
\begin{equation}
    P_{Y|X}^*(\hat{P}_{S,X};\epsilon,r) \defined \operatorname*{arg\,max}_{P_{Y|X}\in \mathcal{D}(\hat{P}_{S,X};\epsilon,r)} \calU_r(\hat{P}_{S,X},P_{Y|X}),
\end{equation}
where $\mathcal{D}(\hat{P}_{S,X};\epsilon,r)$ is the set of all mechanisms $P_{Y|X}$ such that $\mathcal{L}(Q_{S,X},P_{Y|X}) \leq \epsilon$ for all $Q_{S,X}$ with $\| \hat{P}_{S,X} - Q_{S,X}\| \leq r$, and
\begin{equation}
\label{eq:DefRobustPUT}
    \calU_r(\hat{P}_{S,X},P_{Y|X}) \defined \hspace{-5pt} \inf_{Q_{S,X} : \| \hat{P}_{S,X} - Q_{S,X}\| \leq r} \hspace{-5pt} \mathcal{U}(Q_{S,X},P_{Y|X}).
\end{equation}
\end{defn}

For a given privacy mechanism $P_{Y|X}$, the infimum in \eqref{eq:DefRobustPUT} equals the worst case utility attained by $P_{Y|X}$ over all the distributions $Q_{S,X}$ at a distance less than or equal to $r$ from $\hat{P}_{S,X}$. Thus, by definition, $P_{Y|X}^*(\hat{P}_{S,X};\epsilon,r)$ is the privacy mechanism with the best worst-case performance among all privacy mechanisms which ensure an $\epsilon$-privacy guarantee for \emph{all} the distributions at a distance less than or equal to $r$ from $\hat{P}_{S,X}$. In this context, it is natural to investigate the utility degradation incurred by providing such a robust privacy guarantee. For this matter, we introduce the {\it uniform utility-degradation} function as follows.
\begin{defn}
Given $\hat{P}_{S,X}$, $P_{S,X}$, $\epsilon\geq0$, and $r\geq0$, we define
\begin{equation}
\label{eq:DefDelta}
    \Delta(P_{S,X},\hat{P}_{S,X};\epsilon,r) \defined \mathsf{H}(P_{S,X};\epsilon)-\calU(P_{S,X},P_{Y|X}^*),
\end{equation}
where $P_{Y|X}^* = P_{Y|X}^*(\hat{P}_{S,X};\epsilon,r)$.
\end{defn}
Note that when $r=0$, \eqref{eq:DefDelta} measures the utility degradation due to the mismatched estimation; while for $r>0$, \eqref{eq:DefDelta} quantifies the utility degradation incurred by the uniform privacy guarantee and the mismatched estimation.

\subsection{Distance between the Estimated and True Distribution}
\label{Subsection:EmpiricalDistribution}

The distance between the learned and true distributions is the superposition of several errors, for example, estimation error, observation and sampling errors, etc. All these effects can be incorporated into the parameter $r$. Due to space constraints, here we deal only with the estimation error.

A result by Devroye \cite[Lemma~3]{devroye1983equivalence} establishes that, for every $\epsilon\geq\sqrt{20k/n}$,
\begin{equation}
    \Pr\left(\sum_{i=1}^k |V_i - n p_i| > n\epsilon\right) \leq 3 \exp\left(-\frac{n}{25}\epsilon^2\right),
\end{equation}
where $(V_1,\ldots,V_k)$ is a multinomial $(n,p_1,\ldots,p_k)$ random vector. Note that the empirical distribution is a (normalized) multinomial random vector. Hence, by taking $k=M \defined |\mathcal{S}||\mathcal{X}|$ and $\epsilon = \lambda \sqrt{20M/n}$ with $\lambda\geq1$, Devroye's lemma implies that, with probability at least $1-\beta_\lambda$,
\begin{equation}
\label{equa:L1bound_Devroye}
    \|\hat{P}_{S,X} - P_{S,X}\| \leq \lambda \sqrt{\frac{20M}{n}},
\end{equation}
where $\beta_\lambda \defined 3\exp(-4\lambda^2M/5)$ and $\hat{P}_{S,X}$ is the empirical distribution obtained from $\{(s_i,x_i)\}_{i=1}^n$. Even though in this paper we focus on large deviation results, it is worth pointing out that the order $O(\sqrt{M/n})$ is present in other fundamental settings, e.g., the minimax expected loss framework in \cite[Cor.~9]{kamath2015learning}.

\subsection{$f$-Informations and Probability of Correct Guessing}
\label{Subsection:fInformationPc}

We briefly introduce a few definitions that will be used for privacy and utility metrics in the rest of the paper. Let $f:(0,\infty)\to \Reals$ be a convex function with $f(1)=0$. The $f$-divergence between two probability distributions $P$ and $Q$ with $P \ll Q$ is given by \cite{ciszar1967information}
\begin{equation}
D_f(P \| Q)\defined \sum_x Q(x) f\left(\frac{P(x)}{Q(x)}\right).
\end{equation}
More recent developments about the properties of $f$-divergences can be found in \cite{raginsky2016strong,calmon2017strong} and the references therein. With this notation, the $f$-information between two discrete random variables $U$ and $V$ is defined by
\begin{equation}
\begin{aligned}
I_f(P_{U,V}) &\defined D_f(P_{U,V} \| P_U P_V)\\
&=\sum_{u,v} P_{U}(u)P_{V}(v) f\left(\frac{P_{U,V}(u,v)}{P_U(u)P_V(v)}\right).
\end{aligned}
\end{equation}

Also, the probability of correctly guessing $U$, with no additional information, is given by \cite{fehr2014conditional}
\begin{equation}
    P_c(U) \defined \max_{u\in\mathcal{U}} \Pr(U=u).
\end{equation}
Similarly, the probability of correctly guessing $U$ given $V$ is
\begin{equation}
    P_c(U|V) \defined \sum_{v\in\mathcal{V}} \max_{u\in\mathcal{U}} \Pr(U=u,V=v).
\end{equation}
\section{Main Results}

\subsection{Pointwise Privacy Guarantees}
\label{Section:PointwiseRobustness}

Now we study the discrepancy between the guarantees provided for the empirical and true distributions by any fixed mechanism when both privacy and utility are measured using \mbox{$f$-informations}. For space brevity, all the results in this section are stated using the same $f$-information to measure both privacy and utility. It can be shown, {\it mutatis mutandis}, that they hold true also when privacy and utility are measured using different $f$-informations.

The main result of this section is based on the following two technical lemmas. Before stating them, we recall the following definitions. For a given function $g:[0,\infty)\to \Reals$ and $u>0$, we let
\begin{equation}
K_{g,u} \defined \sup \{|g(x)| : x\in[0,u^{-1}]\}.
\end{equation}
The constant $K_{g,u}$ is the so-called supremum norm of $g$ on $[0,u^{-1}]$. In addition, if $g$ is Liptschitz on $[0,u^{-1}]$, we let $L_{g,u}$ be its Lipschitz constant, i.e.,
\begin{equation}
    \min \{L\geq0 : |g(x)-g(y)| \leq L|x-y|, \forall x,y\in[0,u^{-1}]\}.
\end{equation}
A function $g:[0,\infty)\to\mathbb{R}$ is called locally Lipschitz if, for every $t\geq0$, it is Lipschitz on $[0,t]$ with a Lipschitz constant that may depend on $t$. For example, the function $g(x) = x^2$ is locally Lipschitz but not Lipschitz.

\begin{lem}
\label{lem:rob_f_inf1}
Suppose that $S_i \rightarrow X_i \rightarrow Y_i$ for $i=1,2$ and $P_{Y_1|X_1}=P_{Y_2|X_2}$. Let $m_S \defined \min\{P_{S_i}(s) : s\in\mathcal{S},i\in\{1,2\}\}$ and $m_X \defined \min\{P_{X_i}(x) : x\in\mathcal{X},i\in\{1,2\}\}$. For notational simplicity, let
\begin{align}
    \Delta_L &\defined |I_f(P_{S_1,Y_1})-I_f(P_{S_2,Y_2})|,\\
    \Delta_U &\defined |I_f(P_{X_1,Y_1})-I_f(P_{X_2,Y_2})|.
\end{align}
If $f:[0,\infty)\to\Reals$ is locally Lipschitz, then, for all $\delta>0$,
\begin{equation*}
    \Delta_L \leq \begin{cases} A_f |\mathcal{S}| \delta + B_{f,\delta} \|P_{S_1,X_1}-P_{S_2,X_2}\| & m_S < \delta, \\ C_{f,m_S} ||P_{S_1,X_1}-P_{S_2,X_2}|| & \delta \leq m_S,\end{cases}
\end{equation*}
\begin{equation*}
     \Delta_U \leq \begin{cases} A_f |\mathcal{X}| \delta + B_{f,\delta} \|P_{S_1,X_1}-P_{S_2,X_2}\| & m_X < \delta, \\ C_{f,m_X} ||P_{S_1,X_1}-P_{S_2,X_2}|| & \delta \leq m_X,\end{cases}
\end{equation*}
where $A_f= 4 K_{f,m_X}$,
\begin{equation}
    B_{f,\delta} = K_{f,m_X} + 2K_{f,\delta} + (2\delta^{-1} + 1) L_{f,\delta},
\end{equation}
and $C_{f,u} = 2 K_{f,m_X} + (2u^{-1} + 1) L_{f,m_X}$ with $u\in\{m_S,m_X\}$.
\end{lem}


Observe that the previous lemma implicitly assumes that $f(0)=\lim_{x\to0^+} f(x)$ is finite. Examples of $f$-divergences satisfying the assumptions of Lemma~\ref{lem:rob_f_inf1} include the total variation distance, the $\chi^2$-distance, and the Hellinger distance of order $\alpha>1$ (a one-to-one transformation of the R\'{e}nyi divergence of the same order). See \cite{sason2016f} for further examples. Note that, however, KL-divergence cannot be handled by Lemma~\ref{lem:rob_f_inf1} as $|\log(x)| \to \infty$ as $x\to0^+$. Indeed, KL-divergence has a different asymptotic behavior than the one obtained in Theorem~\ref{thm:main_robust} below, see \cite{shamir2010learning}.

Due to limited sample size, not all outcomes of $X$ may be observable in the data set used to design the privacy mechanism, and can significantly impact performance depending on the metric used. Indeed, by taking $P_{S_1,X_1}=P_{S,X}$ and $P_{S_2,X_2} = \hat{P}_{S,X}$, we can see  that the upper bounds for $\Delta_L$ and $\Delta_U$ in Lemma~\ref{lem:rob_f_inf1} become larger as $m_X$ gets smaller. In order to address this issue, we propose  a pre-processing technique which combines the symbols with {\it less} observations. Specifically, for $\gamma\geq0$ and $x_0$ a symbol not belonging to $\calX$, we introduce the pre-processing technique $\Pi_\gamma$ with input alphabet $\calX$ and output alphabet
\begin{equation}
    \calX_\gamma \defined \{x\in\calX : \hat{P}_{X}(x) \geq \gamma\} \cup \{x_0\},
\end{equation}
determined by
\begin{equation}
\label{eq:PreProcessing}
\Pi_\gamma(x) = \begin{cases}x & \hat{P}_{X}(x) \geq \gamma, \\ x_0 & \text{otherwise}. \end{cases}
\end{equation}
Consider the following lemma regarding this pre-processing technique.

\begin{lem}
\label{lem:cut_lesslikely_symbols}
Let $\gamma\geq0$. If $X \to X_0 \to Y_0$ is a Markov chain with $X_0 = \Pi_\gamma(X)$. Then, for every $f$-information,
\begin{equation}
\label{equa:utility_lesslikely}
I_f(P_{X,Y_0})=I_f(P_{X_0,Y_0}).
\end{equation}
\end{lem}


Although this lemma may look counterintuitive at a first glance, its proof relies on the fact that the conditional distributions $P_{Y_0|X}$ and $P_{Y_0|X_0}$ are essentially the same. Specifically,
\begin{equation*}
\begin{aligned}
    &P_{Y_0|X}(y|x) \\
    &= \begin{cases} P_{Y_0|X_0}(y|x) & x\in\{x\in\calX : \hat{P}_{X}(x) \geq \gamma\},\\
    P_{Y_0|X_0}(y|x_0) & x\in\calX\setminus\{x\in\calX : \hat{P}_{X}(x) \geq \gamma\}. \end{cases}
\end{aligned}
\end{equation*}

The following theorem is the main result of this section. It bounds the discrepancy of the privacy-utility guarantees between the learned and true distributions.

\begin{thm}
\label{thm:main_robust}
Let $\gamma\geq0$ and $\hat{P}_{S,X}$ be the empirical distribution of $n$ i.i.d. samples drawn from $P_{S,X}$. Assume that
\begin{equation}
    S\to X\to X_0 \to Y_0,
\end{equation}
where $X_0 = \Pi_\gamma(X)$ and $P_{Y_0|X_0}$ is fixed. Let $P_{S,Y_0}$ and $\hat{P}_{S,Y_0}$ be the joint distributions of $(S,Y_0)$ when the joint distributions of $(S,X)$ are $P_{S,X}$ and $\hat{P}_{S,X}$, respectively. Define $P_{X,Y_0}$ and $\hat{P}_{X,Y_0}$ in an equivalent manner. Let
\begin{align*}
m_S &\defined \min\{\{P_{S}(s) : s\in\mathcal{S}\}\cup\{\hat{P}_{S}(s) : s\in\mathcal{S}\}\},\\
m_X &\defined \min\{\{P_{X_0}(x) : x\in\mathcal{X}_\gamma\}\cup\{\hat{P}_{X_0}(x) : x\in\mathcal{X}_\gamma\}\}.
\end{align*}
If $f:[0,\infty)\to\Reals$ is locally Lipschitz and $m_X \leq m_S$, then, with probability $1-\beta_\lambda$,
\begin{align}
    |I_f(\hat{P}_{S,Y_0})-I_f(P_{S,Y_0})| &\leq C_{f,m_S} \lambda \sqrt{\frac{20M}{n}},\\
    |I_f(\hat{P}_{X_0,Y_0})-I_f(P_{X,Y_0})|
    &\leq  C_{f,m_X} \lambda \sqrt{\frac{20M}{n}},
\end{align}
where $M = |\mathcal{S}||\mathcal{X}|$, $\beta_\lambda = 3\exp(-4\lambda^2M/5)$ with $\lambda \geq 1$ and $C_{f,u}$ is defined in Lemma~\ref{lem:rob_f_inf1}.
\end{thm}

\begin{proof}[\bf Proof of Theorem~\ref{thm:main_robust}]
We first apply Lemma~\ref{lem:rob_f_inf1} with $\delta=m_X$, $P_{S_1,X_1}=P_{S,X_0}$, and $P_{S_2,X_2}=\hat{P}_{S,X_0}$. In particular, we obtain that
\begin{align}
    |I_f(\hat{P}_{S,Y_0})-I_f(P_{S,Y_0})| &\leq C_{f,m_S} \|\hat{P}_{S,X_0}-P_{S,X_0}\|,\\
    |I_f(\hat{P}_{X_0,Y_0})-I_f(P_{X_0,Y_0})|
    &\leq  C_{f,m_X} \|\hat{P}_{S,X_0}-P_{S,X_0}\|.
\end{align}
By the data processing inequality, we have
\begin{equation}
\|\hat{P}_{S,X_0}-P_{S,X_0}\|\leq \|\hat{P}_{S,X}-P_{S,X}\|.
\end{equation}
By the inequality~\eqref{equa:L1bound_Devroye}, with probability at least $1-\beta_\lambda$,
\begin{equation}
\|\hat{P}_{S,X}-P_{S,X}\| \leq \lambda \sqrt{\frac{20M}{n}},
\end{equation}
where $\beta_\lambda = 3\exp(-4\lambda^2M/5)$ and $\lambda \geq 1$.
Hence,
\begin{align}
    |I_f(\hat{P}_{S,Y_0})-I_f(P_{S,Y_0})| &\leq C_{f,m_S} \lambda \sqrt{\frac{20M}{n}},\\
    |I_f(\hat{P}_{X_0,Y_0})-I_f(P_{X_0,Y_0})|
    &\leq  C_{f,m_X} \lambda \sqrt{\frac{20M}{n}}.
\end{align}
By Lemma~\ref{lem:cut_lesslikely_symbols}, we have that
\begin{equation}
I_f(P_{X,Y_0})=I_f(P_{X_0,Y_0}).
\end{equation}
The result follows.
\end{proof}
\subsection{Uniform Privacy Guarantees}
\label{Section:UniformRobustness}

Let $\mathbb{P}$ denote the set of all probability distributions over $\mathcal{S}\times\mathcal{X}$. Assume that prior information about the joint distribution of $S$ and $X$ is available. In this case, we let $\mathbb{Q}\subseteq\mathbb{P}$ be all the joint distributions compatible with the prior knowledge. For a given $\hat{P}_{S,X}\in\mathbb{P}$ and $r\geq0$, we define
\begin{equation}
    \mathbb{Q}_r(\hat{P}_{S,X}) \defined \{Q_{S,X}\in\mathbb{Q} : \|Q_{S,X}-\hat{P}_{S,X}\| \leq r\}.
\end{equation}
In this setting, a natural modification for the uniform privacy mechanism $P_{Y|X}^*$ is the following. Given $\hat{P}_{S,X}\in\mathbb{Q}$, $\epsilon\geq0$, and $r\geq0$, let
\begin{equation}
\label{eq:DefParametricRPM}
    P_{Y|X}^{*}(\hat{P}_{S,X};\epsilon,r) \defined \operatorname*{arg\,max}_{P_{Y|X}\in \mathcal{D}_{\mathbb{Q}}(\hat{P}_{S,X};\epsilon,r)} \calU_r(\hat{P}_{S,X},P_{Y|X})
\end{equation}
where
\begin{equation*}
\mathcal{D}_{\mathbb{Q}}(\hat{P}_{S,X};\epsilon,r) \defined \hspace{-5pt} \bigcap_{Q_{S,X} \in \mathbb{Q}_r(\hat{P}_{S,X})} \hspace{-5pt} \left\{ P_{Y|X} : \mathcal{L}(Q_{S,X},P_{Y|X}) \leq \epsilon\right\},
\end{equation*}
\begin{equation*}
\calU_r(\hat{P}_{S,X},P_{Y|X}) \defined \hspace{-5pt} \inf_{Q_{S,X}\in \mathbb{Q}_r(\hat{P}_{S,X})} \hspace{-5pt} \mathcal{U}(Q_{S,X},P_{Y|X}).
\end{equation*}
Finally, recall that
\begin{equation}
\Delta_{\mathbb{Q}}(P_{S,X},\hat{P}_{S,X};\epsilon,r) \defined \mathsf{H}(P_{S,X};\epsilon) - \calU(P_{S,X},P_{Y|X}^*).
\end{equation}

In order to simplify the notation, in what follows we denote $\hat{P}_{S,X}$, $P_{S,X}$, and $Q_{S,X}$ by $\hat{P}$, $P$ and $Q$, respectively. For the privacy measures under consideration, $f$-information and probability of correct guessing, it has been proved that the optimal privacy mechanism for the PUT in Definition~\ref{defn:PUT} requires an alphabet of size $|\mathcal{X}|+1$, see \cite{witsenhausen1975conditional,asoodeh2017estimation} and references therein. With this in mind, let $\mathbb{F}$ be the set of all row stochastic matrices of dimension $|\mathcal{X}|\times(|\mathcal{X}|+1)$. Note that the set $\mathbb{F}$ models all privacy mechanisms $P_{Y|X}$ with $|\mathcal{Y}| \leq |\mathcal{X}|+1$. In this case, the privacy-utility function in Definition~\ref{defn:PUT} equals
\begin{equation}
\label{eq:ReducedAlphabetH}
    \mathsf{H}(P;\epsilon) = \sup_{\begin{smallmatrix} F\in\mathbb{F} \\ \mathcal{L}(P,F)\leq\epsilon \end{smallmatrix}} \mathcal{U}(P,F),
\end{equation}
for all $P\in\mathbb{P}$. Note that, in principle, the robust privacy mechanism in Definition~\ref{def:robust-utility-privacy}, or the version in \eqref{eq:DefParametricRPM}, may require the use of more than $|\mathcal{X}|+1$ output symbols. However, the following lower bound for $\calU(P,P_{Y|X}^*)$, which can be computed using mechanisms with $|\mathcal{X}|+1$ output symbols, will be enough for our purposes,
\begin{equation}
\label{eq:LowerBoundR}
    \calU(P,P_{Y|X}^*) \geq \sup_{F\in\mathbb{D}_{\mathbb{Q}}(\hat{P};\epsilon,r)} \inf_{Q\in\mathbb{\mathbb{Q}}_r(\hat{P})} \mathcal{U}(Q,F),
\end{equation}
where $\hat{P}\in\mathbb{Q}$ and
\begin{equation}
\mathbb{D}_{\mathbb{Q}}(\hat{P};\epsilon,r) \defined \bigcap_{Q\in\mathbb{Q}_r(\hat{P})} \{F \in \mathbb{F} : \mathcal{L}(Q,F) \leq \epsilon\} \subseteq \mathbb{F}.
\end{equation}
The main result of this section provides an upper bound for $\Delta_{\mathbb{Q}}$ whenever the leakage and utility functions satisfy a H\"{o}lder-like condition. Recall that a function $f:\mathbb{R}\to\mathbb{R}$ is said to be H\"{o}lder continuous of order $\alpha\in[0,1]$ if there exists $K\geq0$ such that $|f(x)-f(y)| \leq K |x-y|^\alpha$ for all $x,y\in\mathbb{R}$.

\begin{thm}
\label{Thm:Delta_gap}
    Assume that $\mathbb{Q}\subseteq\mathbb{P}$ is a closed set and that for every $Q\in\mathbb{Q}$ the functions
    \begin{equation}
    \label{eq:AssumptionsThm2a}
    F \mapsto \mathcal{L}(Q,F) \quad \text{and} \quad F \mapsto \mathcal{U}(Q,F),
    \end{equation}
    are continuous and that \eqref{eq:ReducedAlphabetH} holds true. Furthermore, assume that for a given $\hat{P}\in\mathbb{Q}$ there exist positive constants $r_0$, $\alpha$, $C_L$, and $C_U$ such that
    \begin{equation}
    \label{eq:AssumptionsThm2b}
        |\mathcal{L}(\hat{P},F)-\mathcal{L}(Q,F)| \leq C_L \|\hat{P}-Q\|^\alpha,
    \end{equation}
    \begin{equation}
    \label{eq:AssumptionsThm2c}
        |\mathcal{U}(\hat{P},F)-\mathcal{U}(Q,F)| \leq C_U \|\hat{P}-Q\|^\alpha,
    \end{equation}
    for all $Q\in\mathbb{Q}_{r_0}(\hat{P})$ and all $F\in\mathbb{F}$. If $P\in\mathbb{Q}_{r_0}(\hat{P})$, then, for all $\epsilon>0$ and all $r^\alpha \leq \min\{r_0^\alpha,(\epsilon- \min_F \mathcal{L}(\hat{P},F))/C_L\}$,
    \begin{equation}
        \Delta_{\mathbb{Q}}(P,\hat{P};\epsilon,r)  \leq \mathsf{H}(\hat{P};\epsilon+C_Lr^\alpha)- \mathsf{H}(\hat{P};\epsilon-C_Lr^\alpha) + 2C_Ur^\alpha.
    \end{equation}
\end{thm}


\begin{rem}
Under the assumptions of Theorem~\ref{Thm:Delta_gap}, if $\mathsf{H}(\hat{P};\cdot)$ is Lipschitz continuous with Lipschitz constant $L$ then
\begin{equation}
\label{eq:LipchitzDeltaQ}
    \Delta_{\mathbb{Q}}(P,\hat{P};\epsilon,r) \leq 2(C_U+LC_L)r^\alpha.
\end{equation}
\end{rem}

The assumptions in \eqref{eq:AssumptionsThm2a}, \eqref{eq:AssumptionsThm2b} and \eqref{eq:AssumptionsThm2c} might seem restrictive at a first glance. Nonetheless, as shown in the following, they hold true for our measures of interest: $f$-informations and probability of correct guessing.

\subsubsection{$f$-Divergences}

Assume that both privacy and utility are measured by an $f$-information for a given convex function $f:(0,\infty)\to\mathbb{R}$ with $f(1)=0$, i.e.,
\begin{align}
    \mathcal{L}(Q_{S,X},P_{Y|X}) &\defined I_f(Q_{S,Y}),\\
    \mathcal{U}(Q_{S,X},P_{Y|X}) &\defined I_f(Q_{X,Y}).
\end{align}
A standard convexity argument, see, e.g., \cite{witsenhausen1975conditional}, shows that
\begin{equation}
    \mathsf{H}(Q_{S,X};\epsilon) \defined \sup_{\begin{smallmatrix} S \rightarrow X \rightarrow Y \\ I_f(Q_{S,Y})\leq\epsilon \end{smallmatrix}} I_f(Q_{X,Y})
\end{equation}
admits the expression in \eqref{eq:ReducedAlphabetH}, i.e., it is enough to consider privacy mechanisms taking values on $\mathcal{Y} = \{1,\ldots,|\mathcal{X}|+1\}$.

\begin{example}
\label{Example:RobustnessfDivergence1}
Consider the parametric case where
\begin{align}
    \nonumber \mathbb{Q} \defined & \left\{Q\in\mathbb{P} : \sum_{x\in\mathcal{X}} Q(s,x) \geq \gamma\text{ for all }s\in\mathcal{S}\right\}\\
    & \quad \cap \left\{Q\in\mathbb{P} : \sum_{s\in\mathcal{S}} Q(s,x) \geq \gamma\text{ for all }x\in\mathcal{X}\right\}
\end{align}
for some $\gamma>0$. Note that this corresponds to the case in which $S$ and $X$ have full support and their marginal distributions are bounded away from zero. In this case, Lemma~\ref{lem:rob_f_inf1} implies that the assumptions of Theorem~\ref{Thm:Delta_gap} are satisfied with $r_0=\infty$, $\alpha=1$, and
\begin{equation}
    C_L = C_U = 2 K_{f,\gamma} + (2\gamma^{-1}+1)L_{f,\gamma}.
\end{equation}
In particular, if $f(x) = |x-1|$, then
\begin{equation}
    C_L = C_U \leq 4\gamma^{-1}+1;
\end{equation}
and if $f(x)=x^2-1$, then
\begin{equation}
    C_L = C_U \leq 8\gamma^{-2}.
\end{equation}
\end{example}

\subsubsection{Probability of Correct Guessing}
\label{SubsectionRobustPc}

For ease of notation, let $\mathcal{Y} = \{1,\ldots,|\mathcal{X}|+1\}$. In the setting of the PUT, let
\begin{align}
    \mathcal{L}(Q,F) &= \sum_{y\in\mathcal{Y}} \max_{s\in\mathcal{S}} \sum_{x\in\mathcal{X}} Q(s,x) F(x,y),\\
    \mathcal{U}(Q,F) &= \sum_{y\in\mathcal{Y}} \max_{x\in\mathcal{X}} \sum_{s\in\mathcal{S}} Q(s,x) F(x,y).
\end{align}
The above choice corresponds to the case when the measures of privacy and utility are $P_c(S|Y)$ and $P_c(X|Y)$, respectively. In particular, for each $\epsilon\in[P_c(S),P_c(S|X)]$,
\begin{equation}
    \mathsf{H}(P;\epsilon) = \sup_{\begin{smallmatrix} S \rightarrow X \rightarrow Y \\ P_c(S|Y) \leq \epsilon \end{smallmatrix}} P_c(X|Y).
\end{equation}
This privacy-utility trade-off based on the probability of correctly guessing was recently studied by Asoodeh et al \cite{asoodeh2017estimation}. In this case, it is possible to verify that $\mathcal{L}(P,\cdot)$ and $\mathcal{U}(P,\cdot)$ are continuous.

For $\hat{P},Q\in\mathbb{P}$ and $F\in\mathbb{F}$, let $\Delta_L \defined |\mathcal{L}(\hat{P},F) - \mathcal{L}(Q,F)|$. It can be verified that, for $a_i,b_i\geq0$,
\begin{equation}
    |\max_i a_i - \max_i b_i| \leq \max_i |a_i - b_i|,
\end{equation}
and, in particular,
\begin{align}
     \Delta_L &\leq \sum_{y\in\mathcal{Y}} \left|\max_{s\in\mathcal{S}} (\hat{P}F)(s,y) - \max_{s\in\mathcal{S}} (QF)(s,y)\right|\\
     &\leq \sum_{y\in\mathcal{Y}} \max_{s\in\mathcal{S}} \left| (\hat{P}F)(s,y) - (QF)(s,y)\right|.
\end{align}
Note that $(QF)(s,y) = \sum_{x} Q(s,x) F(x,y)$. Thus, a straightforward manipulation shows that
\begin{align}
     \Delta_L &\leq \sum_{y\in\mathcal{Y}} \max_{s\in\mathcal{S}} \sum_{x\in\mathcal{X}} |\hat{P}(s,x)-Q(s,x)| F(x,y)\\
    &\leq \sum_{s\in\mathcal{S}} \sum_{x\in\mathcal{X}} \sum_{y\in\mathcal{Y}} |\hat{P}(s,x) - Q(s,x)| F(x,y)\\
    &\leq \|\hat{P}-Q\|.
\end{align}
Similarly, it can be shown that
\begin{equation}
    |\mathcal{U}(\hat{P},F) - \mathcal{U}(Q,F)| \leq \| \hat{P} - Q \|.
\end{equation}
Hence, the probability of correct guessing satisfies the assumptions of Theorem~\ref{Thm:Delta_gap} with $r_0=\infty$, $\alpha=1$, $C_L=C_U=1$.

\begin{example}
\label{Example:ParametricPc}
For $p,q\in[0,1]$, we let
\begin{equation}
    p\#q = \left(\begin{matrix} (1-p)(1-q) & (1-p)q \\ pq & p(1-q) \end{matrix}\right),
\end{equation}
and $\mathbb{Q} = \{p\#q : p\in[1/2,1],q\in[0,1/2],p+q\leq1\}$. This selection of $\mathbb{Q}$ captures the case when $S$ is assumed to be a Bernoulli random variable with $\Pr(S=1)=p$ and the channel between $S$ and $X$ is a binary symmetric channel with crossover probability $q$. By Theorem~2 in \cite{asoodeh2017estimation}, for all $Q\in\mathbb{Q}$,
\begin{equation}
    \mathsf{H}(Q;\epsilon) = 1 - \frac{1-q}{p-q} (p+q-2pq) + \epsilon \frac{p+q-2pq}{p-q},
\end{equation}
whenever $\epsilon\in[p,1-q]$. Hence, the bound in \eqref{eq:LipchitzDeltaQ} becomes
\begin{equation}
    \Delta_{\mathbb{Q}}(P,\hat{P};\epsilon,r) \leq \frac{2\hat{p}(1-\hat{q})}{\hat{p}-\hat{q}} r,
\end{equation}
where $\hat{P} \defined \hat{p}\#\hat{q} \in \mathbb{Q}$. By \eqref{equa:L1bound_Devroye}, for $\lambda\geq1$, with probability at least $1-\beta_\lambda$,
\begin{equation}
    \Delta_{\mathbb{Q}}(P_{S,X},\hat{P}_{S,X};\epsilon,4\lambda\sqrt{5/n}) \leq \frac{8\hat{p}(1-\hat{q})}{\hat{p}-\hat{q}} \lambda \sqrt{\frac{5}{n}},
\end{equation}
where $\beta_\lambda \defined 3\exp(-16\lambda^2/5)$.
\end{example}

\bibliographystyle{IEEEtran}
\bibliography{references}

\appendices
\section{Proof of Lemma~\ref{lem:rob_f_inf1}}
\label{Appendix:Proofs}

The following auxiliary lemma will be used in the proof of Lemma~\ref{lem:rob_f_inf1}.

\begin{lem}
\label{lem:auxiliary_mx}
Let $S$, $X$ and $Y$ be random variables supported over finite alphabets $\mathcal{S}$, $\mathcal{X}$ and $\mathcal{Y}$, respectively. Assume that $S\rightarrow X \rightarrow Y$ form a Markov chain in that order. Then, for all $s\in\mathcal{S}$, $x\in\mathcal{X}$ and $y\in\mathcal{Y}$,
\begin{equation}
    \max\left\{\frac{P_{S,Y}(s,y)}{P_S(s)P_Y(y)},\frac{P_{X,Y}(x,y)}{P_X(x)P_Y(y)}\right\} \leq \left(\min_{x\in\mathcal{X}} P_X(x)\right)^{-1}.
\end{equation}
\end{lem}

\begin{proof}
Recall that
\begin{equation}
    \frac{\sum_i a_i}{\sum_i b_i} \leq \max_i \frac{a_i}{b_i},
\end{equation}
whenever $a_i \geq 0$ and $b_i>0$. For a given $y\in\mathcal{Y}$, let $\mathcal{X}_y \defined \{x\in\mathcal{X} : P_{X,Y}(x,y)>0\}$. Note that, given $s\in\mathcal{S}$ and $y\in\mathcal{Y}$,
\begin{align}
\frac{P_{S,Y}(s,y)}{P_S(s)P_Y(y)} &= \frac{\sum_{x\in\mathcal{X}_y} P_{S,X,Y}(s,x,y) }{\sum_{x\in\mathcal{X}_y} P_S(s)P_{X,Y}(x,y)}\\
&\leq \max_{x\in\mathcal{X}_y} \frac{P_{S,X,Y}(s,x,y)}{P_S(s)P_{X,Y}(x,y)}\\
&= \max_{x\in\mathcal{X}_y} \frac{P_{S,X}(s,x)}{P_S(s)P_{X}(x)}\\
&\leq \max_{x\in\mathcal{X}} \frac{1}{P_X(x)} = \left(\min_{x\in\mathcal{X}} P_X(x)\right)^{-1},
\end{align}
where the last inequality follows from the fact that $\mathcal{X}_y\subseteq\mathcal{X}$ and $P_{S,X}(s,x) \leq P_S(s)$ for all $s\in\mathcal{S}$ and $x\in\mathcal{X}$. The rest of the lemma is similar.
\end{proof}

\begin{proof}[\bf Proof of Lemma~\ref{lem:rob_f_inf1}]
First, let's assume that $m_S < \delta$. Let $\mathcal{S}_i = \{s\in\mathcal{S} : P_{S_i}(s) < \delta\}$ for each $i\in\{1,2\}$ and $\mathcal{S}_+ = \mathcal{S}\setminus(\mathcal{S}_1\cup\mathcal{S}_2)$. By the definition of $f$-information and the triangle inequality, we have that
\begin{equation}
    \Delta_L = |I_f(P_{S_1,Y_1})-I_f(P_{S_2,Y_2})| \leq {\rm I} + {\rm II} + {\rm III},
\end{equation}
where
\begin{align}
    {\rm I} =& \sum_{s\in\mathcal{S}_1} \sum_{y\in\mathcal{Y}}P_{S_1}(s) P_{Y_1}(y) \left|f\left(\frac{P_{S_1,Y_1}(s,y)}{P_{S_1}(s)P_{Y_1}(y)}\right)\right|\\
    &+\sum_{s\in\mathcal{S}_1} \sum_{y\in\mathcal{Y}}P_{S_2}(s) P_{Y_2}(y) \left|f\left(\frac{P_{S_2,Y_2}(s,y)}{P_{S_2}(s)P_{Y_2}(y)}\right)\right|,
\end{align}
\begin{align}
    {\rm II} =& \sum_{s\in\mathcal{S}_2} \sum_{y\in\mathcal{Y}} P_{S_1}(s) P_{Y_1}(y) \left|f\left(\frac{P_{S_1,Y_1}(s,y)}{P_{S_1}(s)P_{Y_1}(y)}\right)\right|\\
    &+\sum_{s\in\mathcal{S}_2} \sum_{y\in\mathcal{Y}} P_{S_2}(s) P_{Y_2}(y) \left|f\left(\frac{P_{S_2,Y_2}(s,y)}{P_{S_2}(s)P_{Y_2}(y)}\right)\right|,
\end{align}
\begin{align}
    {\rm III} = \bigg| \sum_{s\in\mathcal{S}_+} \sum_{y\in\mathcal{Y}} &\left(P_{S_1}(s) P_{Y_1}(y) f\left(\frac{P_{S_1,Y_1}(s,y)}{P_{S_1}(s)P_{Y_1}(y)}\right) \right. \\
    & \left.- P_{S_2}(s) P_{Y_2}(y) f\left(\frac{P_{S_2,Y_2}(s,y)}{P_{S_2}(s)P_{Y_2}(y)}\right) \right)\bigg|.
\end{align}
By Lemma~\ref{lem:auxiliary_mx}, we have that
\begin{equation}
    \max\left\{\frac{P_{S_1,Y_1}(s,y)}{P_{S_1}(s)P_{Y_1}(y)},\frac{P_{S_2,Y_2}(s,y)}{P_{S_2}(s)P_{Y_2}(y)}\right\} \leq m_X^{-1}.
\end{equation}
In particular, we have that
\begin{equation}
\label{equa:BoundA}
    {\rm I} \leq K_{f,m_X} (P_{S_1}(\mathcal{S}_1)+P_{S_2}(\mathcal{S}_1)).
\end{equation}
Since $P_{S_1}(\mathcal{S}_1) + P_{S_2}(\mathcal{S}_1) = P_{S_2}(\mathcal{S}_1)-P_{S_1}(\mathcal{S}_1) + 2 P_{S_1}(\mathcal{S}_1)$, the definition of $\mathcal{S}_1$ implies that
\begin{equation}
\label{equa:BoundA1}
    P_{S_1}(\mathcal{S}_1) + P_{S_2}(\mathcal{S}_1) \leq \frac{1}{2}\|P_{S_1}-P_{S_2}\| + 2 |\mathcal{S}| \delta.
\end{equation}
Note that
\begin{equation}
\label{equa:DPITV}
    \max\{\|P_{S_1}-P_{S_2}\|,\|P_{X_1}-P_{X_2}\|\} \leq \|P_{S_1,X_1} - P_{S_2,X_2}\|.
\end{equation}
Hence, \eqref{equa:BoundA} and \eqref{equa:BoundA1} lead to
\begin{equation}
    {\rm I} \leq 2 K_{f,m_X} |\mathcal{S}| \delta + \frac{K_{f,m_X}}{2} \|P_{S_1,X_1}-P_{S_2,X_2}\|.
\end{equation}
Using a similar argument, we conclude that
\begin{equation}
    {\rm II} \leq 2 K_{f,m_X} |\mathcal{S}| \delta + \frac{K_{f,m_X}}{2} \|P_{S_1,X_1}-P_{S_2,X_2}\|.
\end{equation}
By the triangle inequality ${\rm III} \leq {\rm III}_1 + {\rm III}_2$, where
\begin{align}
    \nonumber {\rm III}_1 = \sum_{s\in\mathcal{S}_+} \sum_{y\in\mathcal{Y}} &|P_{S_1}(s) P_{Y_1}(y) - P_{S_2}(s) P_{Y_2}(y)|\\
    \label{eq:DefIII1} &\times\left|f\left(\frac{P_{S_1,Y_1}(s,y)}{P_{S_1}(s)P_{Y_1}(y)}\right)\right|,
\end{align}
\begin{align}
    \nonumber {\rm III}_2 = \sum_{s\in\mathcal{S}_+} &\sum_{y\in\mathcal{Y}} P_{S_2}(s) P_{Y_2}(y)\\ \label{eq:DefIII2} &\times\left|f\left(\frac{P_{S_1,Y_1}(s,y)}{P_{S_1}(s)P_{Y_1}(y)}\right) - f\left(\frac{P_{S_2,Y_2}(s,y)}{P_{S_2}(s)P_{Y_2}(y)}\right)\right|.
\end{align}
By the definition of $\mathcal{S}_+$, we have that, for all $s\in\mathcal{S}_+$ and $y\in\mathcal{Y}$,
\begin{align}
\frac{P_{S_1,Y_1}(s,y)}{P_{S_1}(s)P_{Y_1}(y)} &= \frac{P_{S_1|Y_1}(s|y)}{P_{S_1}(s)}\\
&\leq \frac{1}{P_{S_1}(s)} \leq \delta^{-1}.
\end{align}
Recall that $|f(x)|\leq K_{f,\delta}$ for all $x\in [0,\delta^{-1}]$. Hence,
\begin{equation}
\begin{aligned}
    {\rm III}_1 
    &\leq K_{f,\delta} \sum_{s\in\mathcal{S}_+} \sum_{y\in\mathcal{Y}} |P_{S_1}(s) P_{Y_1}(y) - P_{S_2}(s) P_{Y_2}(y)|\\
    &\leq K_{f,\delta} \sum_{s\in\mathcal{S}} \sum_{y\in\mathcal{Y}} \left( P_{Y_1}(y)|P_{S_1}(s) - P_{S_2}(s)|\right.\\
    &\quad \quad \quad \quad \quad \quad \quad \quad\left.
    +P_{S_2}(s)|P_{Y_1}(y)-P_{Y_2}(y)|\right)\\
    &= K_{f,\delta} \left(\| P_{S_1} - P_{S_2} \| + \| P_{Y_1} - P_{Y_2}\|\right).
\end{aligned}
\end{equation}
The data processing inequality implies that
\begin{equation}
\|P_{Y_1}-P_{Y_2}\| \leq \|P_{X_1}-P_{X_2}\|,
\end{equation}
and by \eqref{equa:DPITV} we obtain that
\begin{equation}
    {\rm III}_1 \leq 2 K_{f,\delta} \| P_{S_1,X_1} - P_{S_2,X_2} \|.
\end{equation}
Similarly, for all $s\in\mathcal{S}_+$ and $y\in\mathcal{Y}$ we have that
\begin{equation}
    \max \left\{\frac{P_{S_1,Y_1}(s,y)}{P_{S_1}(s)P_{Y_1}(y)},\frac{P_{S_2,Y_2}(s,y)}{P_{S_2}(s)P_{Y_2}(y)}\right\} \leq \delta^{-1}.
\end{equation}
Recall that $f$ is Lipschitz on $[0,\delta^{-1}]$ and $L_{f,\delta}$ is its Lipschitz constant. In particular,
\begin{equation}
\label{equa:BoundD}
\begin{aligned}
    {\rm III}_2 \leq L_{f,\delta} \sum_{s\in\mathcal{S}_+} \sum_{y\in\mathcal{Y}} &\frac{1}{P_{S_1}(s)P_{Y_1}(y)}\\ \times&|P_{S_2}(s)P_{Y_2}(y)P_{S_1,Y_1}(s,y)\\
    &-P_{S_1}(s)P_{Y_1}(y)P_{S_2,Y_2}(s,y)|.
\end{aligned}
\end{equation}
By the triangle inequality,
\begin{equation}
|P_{S_2}(s)P_{Y_2}(y)P_{S_1,Y_1}(s,y)-P_{S_1}(s)P_{Y_1}(y)P_{S_2,Y_2}(s,y)|
\end{equation}
is upper bounded by
\begin{equation}
\begin{aligned}
    &P_{S_1,Y_1}(s,y) |P_{S_2}(s)P_{Y_2}(y)-P_{S_1}(s)P_{Y_1}(y)|\\
    & \quad + P_{S_1}(s)P_{Y_1}(y) |P_{S_1,Y_1}(s,y)-P_{S_2,Y_2}(s,y)|.
\end{aligned}
\end{equation}
Since $P_{S_1,Y_1}(s,y) \leq P_{Y_1}(y)$ for all $s\in\mathcal{S}$ and $y\in\mathcal{Y}$, \eqref{equa:BoundD} leads to
\begin{equation}
\begin{aligned}
    {\rm III}_2 
    &\leq L_{f,\delta} \sum_{s\in\mathcal{S}_+} \sum_{y\in\mathcal{Y}} \frac{1}{P_{S_1}(s)} |P_{S_2}(s)P_{Y_2}(y)-P_{S_1}(s)P_{Y_1}(y)|\\
    &\quad 
    + L_{f,\delta} \sum_{s\in\mathcal{S}_+} \sum_{y\in\mathcal{Y}} |P_{S_1,Y_1}(s,y)-P_{S_2,Y_2}(s,y)|\\
    &\leq \delta^{-1} L_{f,\delta} \Big(\|P_{Y_1}-P_{Y_2}\|+\|P_{S_1}-P_{S_2}\|\Big)\\
    &\quad + L_{f,\delta}\|P_{S_1,Y_1}-P_{S_2,Y_2}\|.
\end{aligned}
\end{equation}
By assumption, $P_{Y_1|X_1} = P_{Y_2|X_2}$ and hence
\begin{equation}
\|P_{S_1,Y_1} - P_{S_2,Y_2}\| \leq \|P_{S_1,X_1} - P_{S_2,X_2}\|.
\end{equation}
Therefore,
\begin{equation}
    {\rm III}_2 
    \leq \left(2\delta^{-1} + 1\right) L_{f,\delta} \|P_{S_1,X_1} - P_{S_2,X_2}\|.
\end{equation}
Since $\Delta_L \leq {\rm I} + {\rm II} + {\rm III}_1 + {\rm III}_2$, we obtain the upper bound
\begin{equation}
    \Delta_L \leq 4 K_{f,m_X} |\mathcal{S}| \delta + B_{f,\delta} \|P_{S_1,X_1} - P_{S_2,X_2}\|,
\end{equation}
where $B_{f,\delta} = K_{f,m_X} + 2K_{f,\delta} + (2\delta^{-1} + 1) L_{f,\delta}$.

Now assume that $\delta \leq m_S$. In particular, we have that $\mathcal{S}_1 = \mathcal{S}_2 = \emptyset$ and hence $\Delta_L \leq {\rm III}_1 + {\rm III}_2$ with ${\rm III}_1$ and ${\rm III}_2$ defined as in \eqref{eq:DefIII1} and \eqref{eq:DefIII2}, respectively. By Lemma~\ref{lem:auxiliary_mx},
\begin{equation}
\max\left\{\frac{P_{S_1,Y_1}(s,y)}{P_{S_1}(s)P_{Y_1}(y)},\frac{P_{S_2,Y_2}(s,y)}{P_{S_2}(s)P_{Y_2}(y)}\right\} \leq m_X^{-1}.
\end{equation}
In particular, we have that
\begin{equation}
    {\rm III}_1 \leq 2 K_{f,m_X} \|P_{S_1,X_1} - P_{S_2,X_2}\|.
\end{equation}
\begin{equation}
    {\rm III}_2 
    \leq \left(2m_S^{-1} + 1\right) L_{f,m_X} \|P_{S_1,X_1} - P_{S_2,X_2}\|.
\end{equation}
Hence,
\begin{equation}
    \Delta_L \leq C_{f,m_S} \|P_{S_1,X_1} - P_{S_2,X_2}\|,
\end{equation}
where $C_{f,m_S} = 2 K_{f,m_X} + (2m_S^{-1} + 1) L_{f,m_X}$.

Mutatis mutandis, setting $\mathcal{X}_i = \{x\in\mathcal{X} : P_{X_i}(x) < \delta\}$ for each $i\in\{1,2\}$ and $\mathcal{X}_+ = \mathcal{X}\setminus(\mathcal{X}_1\cup\mathcal{X}_2)$, it can be shown that
\begin{equation}
    \Delta_U \leq 4 K_{f,m_X} |\mathcal{X}| \delta + B_{f,\delta} \|P_{S_1,X_1}-P_{S_2,X_2}\|,
\end{equation}
when $m_X < \delta$. If $\delta \leq m_X$, it can be shown that
\begin{equation}
    \Delta_U \leq C_{f,m_X} \|P_{S_1,X_1} - P_{S_2,X_2}\|,
\end{equation}
where $C_{f,m_X} = 2 K_{f,m_X} + (2m_X^{-1} + 1) L_{f,m_X}$.
\end{proof}
\section{Proof of Lemma~\ref{lem:cut_lesslikely_symbols}}
\label{Appendix:Proofs_lem:cut_lesslikely_symbols}

\begin{proof}[\bf Proof of Lemma~\ref{lem:cut_lesslikely_symbols}]
First, we define $\mathcal{X}_1\defined \{x\in\calX : \hat{P}_{X}(x) \geq \gamma\}$. By the construction of $X_0$, we have that $P_{X_0|X}(x'|x) = 1$ whenever $x\in\calX_1$ and $x'=x$, or $x\in\calX_1^c$ and $x'=x_0$; in all other cases $P_{X_0|X}(x'|x) = 0$. In particular, 
\begin{equation}
\label{eq:ProofLemma2Marginals}
P_{X_0}(x_0)=\sum_{x\in\mathcal{X}_1^c} P_X(x) \text{ and } P_{X_0}(x) = P_X(x) \text{ for } x\in\calX_1.
\end{equation}
By the law of total probability and Markovianity
\begin{equation}
    P_{X,Y_0}(x,y) = \sum_{x'\in\calX_1\cup\{x_0\}} P_X(x) P_{X_0|X}(x'|x) P_{Y_0|X_0}(y|x'),
\end{equation}
for all $x\in\calX$ and $y\in\calY$. In particular, for all $x\in\calX_1$,
\begin{equation}
\label{eq:ProofLemma2PartI}
    P_{X,Y_0}(x,y) = P_X(x) P_{Y_0|X_0}(y|x).
\end{equation}
Similarly, for $x\in\calX_1^c$,
\begin{equation}
\label{eq:ProofLemma2PartII}
    P_{X,Y_0}(x,y) = P_X(x) P_{Y_0|X_0}(y|x_0).
\end{equation}
By the definition of $f$-information, \eqref{eq:ProofLemma2Marginals}, \eqref{eq:ProofLemma2PartI}, and \eqref{eq:ProofLemma2PartII},
\begin{align}
I_f(P_{X,Y_0}) = &\sum_{x\in \mathcal{X}_1} \sum_{y\in \mathcal{Y}} P_{X}(x)P_{Y_0}(y)f\left(\frac{P_{X,Y_0}(x,y)}{P_{X}(x)P_{Y_0}(y)}\right)\\
&+\sum_{x\in \mathcal{X}_1^c} \sum_{y\in \mathcal{Y}} P_{X}(x)P_{Y_0}(y)f\left(\frac{P_{X,Y_0}(x,y)}{P_{X}(x)P_{Y_0}(y)}\right)\\
= &\sum_{x\in \mathcal{X}_1} \sum_{y\in \mathcal{Y}} P_{X_0}(x)P_{Y_0}(y)f\left(\frac{P_{X_0,Y_0}(x,y)}{P_{X_0}(x)P_{Y_0}(y)}\right)\\
&+ \sum_{y\in \mathcal{Y}} P_{X_0}(x_0)P_{Y_0}(y)f\left(\frac{P_{X_0,Y_0}(x_0,y)}{P_{X_0}(x_0)P_{Y_0}(y)}\right)\\
=& I_f(P_{X_0,Y_0}),
\end{align}
as required.
\end{proof}
\section{Proof of Theorem~\ref{Thm:Delta_gap}}
\label{Appendix:ProofRobustness}

\begin{proof}[\bf Proof of Theorem~\ref{Thm:Delta_gap}]
First we show that, for all $r$ and $\epsilon$ with $P \in \mathbb{Q}_r(\hat{P})$,
\begin{equation}
\label{equa:bound_M_1}
\mathsf{H}(P;\epsilon) \leq \mathsf{H}(\hat{P};\epsilon+C_L r^\alpha) + C_U r^\alpha.
\end{equation}
Since, for fixed $P$, both $\mathcal{L}(P,\cdot)$ and $\mathcal{U}(P,\cdot)$ are continuous, there exists $F\in\mathbb{F}$ such that $\mathcal{L}(P,F) \leq \epsilon$ and 
\begin{equation}
\mathsf{H}(P;\epsilon) = \mathcal{U}(P,F).
\end{equation}
By assumption, we have that
\begin{equation}
|\mathcal{U}(\hat{P},F)-\mathcal{U}(P,F)| \leq C_U \|\hat{P}-P\|^\alpha \leq C_U r^\alpha.
\end{equation}
In particular,
\begin{equation}
\label{equa:bound_M}
\mathsf{H}(P;\epsilon) \leq \mathcal{U}(\hat{P},F) + C_U r^\alpha.
\end{equation}
Similarly, since
\begin{equation}
|\mathcal{L}(\hat{P},F)-\mathcal{L}(P,F)| \leq C_L\|\hat{P}-P\|^\alpha \leq C_L r^\alpha,
\end{equation}
we have
\begin{equation}
\label{equa:bound_F_P}
\mathcal{L}(\hat{P},F) \leq \mathcal{L}(P,F)+ C_L r^\alpha \leq \epsilon + C_L r^\alpha.
\end{equation}
Therefore, from inequality (\ref{equa:bound_M}) and Definition \ref{defn:PUT}, we have 
\begin{equation}
\mathsf{H}(P;\epsilon) \leq \mathsf{H}(\hat{P};\epsilon + C_L r^\alpha) + C_U r^\alpha.
\end{equation}

Next, we prove that
\begin{equation}
\label{equa:bound_R_2}
\mathsf{H}(\hat{P};\epsilon - C_Lr^\alpha)-C_Ur^\alpha \leq \mathcal{U}(P,F^*)
\end{equation}
where we denote $P_{Y|X}^{*}$ by $F^*$. Let $F_0\in\mathbb{F}$ be such that $\mathcal{L}(\hat{P},F_0)\leq \epsilon-C_Lr^\alpha$ and
\begin{equation}
\mathcal{U}(\hat{P},F_0)=\mathsf{H}(\hat{P};\epsilon-C_Lr^\alpha).
\end{equation}
Since for a fixed $\hat{P}$, both $\mathcal{U}(\hat{P},\cdot)$ and $\mathcal{L}(\hat{P},\cdot)$ are continuous, there exists at least one such $F_0$. By assumption, we have for any $Q\in\mathbb{Q}_r(\hat{P})$,
\begin{equation}
|\mathcal{L}(\hat{P},F_0)-\mathcal{L}(Q,F_0)|\leq C_L\|\hat{P}-Q\|^\alpha \leq C_Lr^\alpha.
\end{equation}
In particular, we have that
$\mathcal{L}(Q,F_0) \leq \mathcal{L}(\hat{P},F_0)+C_L r^\alpha \leq \epsilon$. Hence, $F_0\in \mathbb{D}_{\mathbb{Q}}(\hat{P};\epsilon,r)$ and, from the lower bound in \eqref{eq:LowerBoundR},
\begin{equation}
\label{eq:InqThm22}
\inf_{Q\in \mathbb{Q}_r(\hat{P})} \mathcal{U}(Q,F_0) \leq  \mathcal{U}(P,F^*).
\end{equation}
Since, by assumption, 
\begin{equation}
|\mathcal{U}(\hat{P},F_0)-\mathcal{U}(Q,F_0)| \leq C_U \|\hat{P}-Q\|^\alpha\leq C_Ur^\alpha,
\end{equation} 
we have that 
\begin{align}
\mathcal{U}(Q,F_0) &\geq \mathcal{U}(\hat{P},F_0)-C_Ur^\alpha\\
&= \mathsf{H}(\hat{P};\epsilon-C_Lr^\alpha)-C_Ur^\alpha.
\end{align}
In particular, this implies that
\begin{equation}
\label{eq:InqThm21}
\inf_{Q\in \mathbb{Q}_r(\hat{P})} \mathcal{U}(Q,F_0) \geq \mathsf{H}(\hat{P};\epsilon-C_Lr^\alpha)-C_U r^\alpha.
\end{equation}
Combining \eqref{eq:InqThm21} and \eqref{eq:InqThm22}, inequality~(\ref{equa:bound_R_2}) holds. Combining (\ref{equa:bound_M_1}) and (\ref{equa:bound_R_2}) together, we get the desired conclusion.
\end{proof}
\end{document}